\documentclass[11pt,letter]{article}
\pdfoutput=1

\usepackage[T1]{fontenc}
\usepackage{microtype}
\usepackage{fullpage}
\usepackage{authblk}
\usepackage{amsthm}
\usepackage{amssymb}
\usepackage[english]{babel}
\usepackage[utf8]{inputenc}
\usepackage{amsmath}
\usepackage{amsfonts}
\usepackage{graphicx}
\usepackage[colorinlistoftodos]{todonotes}
\usepackage{mathtools}
\usepackage{enumerate}
\usepackage{enumitem}
\usepackage{lineno}
\usepackage{thmtools,thm-restate}
\usepackage{booktabs}
\usepackage{multirow}
\usepackage{makecell}
\usepackage{geometry}
\usepackage{times}
\usepackage{tikz}
\usetikzlibrary{arrows,shapes,positioning,tikzmark,calc}

\usepackage{color}
\definecolor{darkgreen}{rgb}{0,0.5,0}
\definecolor{darkblue}{rgb}{0,0,0.8}
\definecolor{darkred}{rgb}{0.8,0,0}

\usepackage{hyperref}
\hypersetup{
   unicode=false,          % non-Latin characters in Acrobat’s bookmarks
   colorlinks=true,        % false: boxed links; true: colored links
   linkcolor=darkred,          % color of internal links (change box color with linkbordercolor)
   citecolor=darkblue,        % color of links to bibliography
   filecolor=magenta,      % color of file links
   urlcolor=black           % color of external links
}

\newtheorem{definition}{Definition}[section]
\newtheorem{lemma}[definition]{Lemma}
\newtheorem{theorem}[definition]{Theorem}

\newtheorem{observation}[definition]{Observation}

\newcommand{\bigo}{\mathcal{O}}
\newcommand{\ignore}[1]{}

\newcommand{\N}{\mathbb N}
\newcommand{\diam}{\mathrm{diam}}
\newcommand{\RS}{\mathsf{RS}}
\newcommand{\lab}{\mathsf{label}}

\title{\textbf{Hardness of exact distance queries in sparse graphs\\ through hub labeling}}

\author[1]{\Large Adrian Kosowski}
\author[2]{Przemys\l{}aw~Uzna\'nski}
\author[1]{Laurent Viennot}
\affil[1]{Inria, Paris University, France\footnote{Supported by Irif CNRS laboratory and ANR projects DESCARTES (ANR-16-CE40-0023), DISTANCIA (ANR-17-CE40-0015) and MULTIMOD (ANR-17-CE22-0016).}}
\affil[2]{Institute of Computer Science, University of Wrocław, Poland\footnote{Supported by Polish National Science Centre grant 2016/22/E/ST6/00499.}}
\date{}

\begin{document}
\maketitle

\thispagestyle{empty}

\begin{abstract}
A \emph{distance labeling scheme} is an assignment of bit-labels to the vertices of an undirected, unweighted graph such that the distance between any pair of vertices can be decoded solely from their labels. An important class of distance labeling schemes is that of \emph{hub labelings}, where a node $v \in G$ stores its distance to the so-called hubs $S_v \subseteq V$, chosen so that for any $u,v \in V$ there is $w \in S_u \cap S_v$ belonging to some shortest $uv$ path. Notice that for most existing graph classes, the best distance labelling constructions existing use at some point a hub labeling scheme at least as a key building block. %Primary complexity measure is the size of the hubsets.

Our interest lies in hub labelings of sparse graphs, i.e., those with $|E(G)| = \bigo(n)$, for which we show a lowerbound of $\frac{n}{2^{\bigo(\sqrt{\log n})}}$ for the average size of the hubsets. Additionally, we show a hub-labeling construction for sparse graphs of average size $\bigo(\frac{n}{\RS(n)^{c}})$ for some $0 < c < 1$, where $\RS(n)$ is the so-called Ruzsa-Szemer{\'e}di function, linked to structure of induced matchings in dense graphs. This implies that further improving the lower bound on hub labeling size to $\frac{n}{2^{(\log n)^{o(1)}}}$ would require a breakthrough in the study of lower bounds on $\RS(n)$, which have resisted substantial improvement in the last 70 years.

For general distance labeling of sparse graphs, we show a lowerbound of $\frac{1}{2^{\Theta(\sqrt{\log n})}}\textsc{SumIndex}(n)$, where $\textsc{SumIndex}(n)$ is the communication complexity of the Sum-Index problem over $Z_n$. Our results suggest that the best achievable hub-label size and distance-label size in sparse graphs may be $\Theta(\frac{n}{2^{(\log n)^c}})$ for some $0<c < 1$.

\end{abstract}
%\vfill
%\clearpage

\section{Introduction}
Hub labeling schemes are a popular way of distributed encoding shortest path structure for easy retrieval. Given a graph $G = (V,E)$ (across this paper we assume, unless stated otherwise, that $G$ is undirected and unweighted, with $|V| = n$ and $|E|=m$), we store with each vertex $v \in V(G)$ the so-called hubset $S_v \subseteq V(G)$ together with shortest path distances between $v$ and its hubs. The distance query $uv$ is resolved by returning $$\min_{w \in S(u)\cap S(v)} \textrm{dist}(u,w) + \textrm{dist}(w,v).$$
This framework was introduced in \cite{Cohen:2003:RDQ:942270.944300} under the name of 2-hop covers, further explored by \cite{Abraham11onapproximate}) (as landmark labelings) and \cite{Abraham:2012:HHL:2404160.2404164}. The computed distance between all pairs of nodes $u$ and $v$ is exact if set $S(u)\cap S(v)$ contains at least one node on some shortest $u-v$ path. This property of the family of sets $(S(u) : u\in V)$ is known as \emph{shortest path cover}. 

Hub labeling schemes are a special case of the more general \emph{distance labeling} schemes,
in which the task is the assignment of a binary string $\lab(u)$ to each node $u\in V$, so that the graph distance between $u$ and $v$ is uniquely determined by the pair of labels: $\lab(u)$ and $\lab(v)$, using a decoding function specified as part of the distance labeling. A central question in the distributed computing literature consists in identifying, for a given graph class, the optimal size of a distance labeling, expressed in terms on a bound on the average or maximum size of the labels used, taken over all nodes of a graph, for graphs belonging to the considered graph class.

In general graphs, the distance labeling problem was first investigated by Graham and Pollak~\cite{pollak}, who provided the first labeling scheme with labels of size $\bigo(n)$. A subsequent line of research contributed to reducing the \emph{decoding time}, with results of \cite{Gavoille:2004:DLG:1036161.1036165, WP11} and finally \cite{DBLP:conf/soda/AlstrupGHP16}  presented a scheme with labels of size $\frac{\log_2 3}{2} n + o(n)$ and $\bigo(1)$ decoding time. This result is asymptotically tight: by a simple counting argument, the label size of a node has to be at least $\frac{1}{2} n - O(1)$, regardless of decoding time~\cite{Gavoille:2004:DLG:1036161.1036165}. 

By contrast, when considering classes of \emph{sparse graphs} (where $m = \bigo(n)$), no similar asymptotically tight bounds on distance label size are known. The first sublinear-space distance labeling schemes were proposed only recently: \cite{Sublinear} and later \cite{DBLP:conf/wdag/GawrychowskiKU16}  presented a hub labeling scheme with hubs of size $\bigo(\frac{n}{\log n} \log \log n)$ and $\bigo(\frac{n}{\log n})$, respectively. Interpreted in terms of distance labeling schemes, these works provide $\bigo(\frac{n}{\log n} (\log \log n)^2)$ and $\bigo(\frac{n}{\log n} \log \log n)$ bits per label respectively, through some careful encoding of distances to vertices of hubsets.
This leaves a substantial gap with respect to the best known lower bound of $\Omega(\sqrt{n})$ for the bit size of distance labels and of hub labeling schemes per node, given in \cite{Gavoille:2004:DLG:1036161.1036165}. 

The objective of this paper is to help understand the source of hardness of distance labeling in sparse graphs. We show a basic obstacle on the path to the construction of better distance labelings: a lower bound of  $\frac{n}{2^{\bigo(\sqrt{\log n})}}$ on the hub set size in hub labeling. As we discuss further on in the related work section, this is an issue as all currently known approaches to distance labeling rely inherently (possibly implicitly) on the construction of some form of hub labels.
More rigorously, we succeed in tying-in the distance label size to two long-standing open questions: one from combinatorics on the value of the so-called Ruzsa-Szemerédi function $\RS(n)$ on graphs, one from communication complexity on the bit complexity $\textsc{SumIndex}(n)$ of a basic 3-party communication problem known as Sum-Index. 

We remark that our results can perhaps help to shed further light on the question of designing (centralized) distance oracles for sparse graphs. A natural objective, so far unachieved in general, would be to obtain for $n$-node sparse graphs a spectrum of data structures, using space $S$ and resolving exact distance queries in time $T$, with a time-space tradeoff of $ST = \tilde O(n^2)$. This tradeoff is trivially achieved at its endpoints ($S = \tilde O(n)$, $S = \tilde O(n^2)$), but the existence of such oracles appears open, e.g., for $S = \tilde O(n^{3/2})$ and $T = O(n^{1/2})$ (cf.~\cite{DBLP:journals/corr/abs-1006-1117Cohen,sommer2009distance}). Our result precludes the existence of such a centralized oracle relying on an application of hub labeling. 

\subsection{Related work}

\paragraph{Distance labelings vs. Hub labelings.}

It is noteworthy that for most existing graph classes, the best known distance labelling constructions are based on hub labeling schemes, either explicitly or as a key building block. Such constructions usually involve some form of compression and/or encoding of all distances (from a vertex to its hubs), to avoid $\log n$ overhead when going from hubsets to binary labels (see e.g.~\cite{DBLP:conf/wdag/GawrychowskiKU16}).

For example, for the case of arbitrary graphs, a distance labeling may be constructed as follows: an additive approximation scheme for hub-labeling is constructed, that is for each pair $uv$, there is $w \in S(u) \cap S(v)$ such that either $w$  or some neighbor $x \in N(w)$ is on shortest $uv$ path. This guarantees that the absolute error of estimation is either $0,1$ or $2$. Constructing such (small) approximate hub-set and complementing it with explicit correction tables (which require $\frac{\log_2 3}{2} n$ bits per vertex) suffices. Many more ingredients and insights are required to achieve constant decoding time as described in~\cite{DBLP:conf/soda/AlstrupGHP16}, which goes beyond the scope of this overview.

\paragraph{Distance labeling of sparse graphs and sub-classes.}

Constructions for distance labeling in sparse graphs were considered in the previously mentioned papers \cite{Sublinear} and  \cite{DBLP:conf/wdag/GawrychowskiKU16}. Both of them use hubsets as the underlying technique. The first of these constructions relies on the observation that, by selecting randomly a hubset $S$ of size $\bigo(\frac{n}{D} \log D)$ and assigning $S \subseteq S(u)$ for any $u$, one already covers almost all pairs of vertices which are at distance $D$ to each other, except at most a $1/D$ fraction of these pairs. Selecting $D = \Theta(\log n)$ and storing vertices closer than $D$ as hubs explicitly leads to desired hubset sizes (however, a careful approach is required when dealing with large degree vertices which may appear in a sparse graph, guaranteed only to have constant average degree). Second work shaves $\log \log n$ term by derandomizing the construction. 

For the class of trees, the constructions of \cite{DBLP:journals/jgt/Peleg00} and \cite{DBLP:conf/icalp/AlstrupGHP16} are based on  selection of central vertices as hubs, and proceeding recursively on obtained subtrees. In terms of bit size, those constructions give $\Theta(\log^2 n)$ bits per label (which is asymptotically optimal due to lowerbound of \cite{Gavoille:2004:DLG:1036161.1036165}), which corresponds to $\log n$ hubs per vertex. The construction of \cite{DBLP:conf/podc/FreedmanGNW17} through very careful assignment of pieces of information thorough the construction achieves  $\frac{1}{4} \log^2 n + o(\log^2 n)$ bits per label, which is optimal up to lower order terms due to the lowerbound of \cite{DBLP:conf/icalp/AlstrupGHP16}. 

For planar graphs, the main technical ingredient is an existence of small size separators. Specifically, in any planar graph there is a $S \subseteq V$ such that $|S| = \bigo(\sqrt{n})$ and $S$ separates $V$ into $V_1$ and $V_2$ that are balanced (up to a factor 2) in sizes. Taking advantage of this, \cite{Gavoille:2004:DLG:1036161.1036165} described $\bigo(\sqrt{n})$ hub labeling and $\bigo(\sqrt{n} \log n)$ bits distance labeling schemes, by applying the separation recursively. They also present  $\Omega(\sqrt[3]{n})$ lower bound (for both hub and distance labeling schemes). The distance labeling size was later improved in \cite{DBLP:journals/corr/GawrychowskiU16} to $\bigo(\sqrt{n})$, through better encoding of distances from vertex to its hub (and thus avoiding $\log n$ overhead).

\paragraph{Lower bounds.}
The aforementioned lowerbounds on size of distance labels and/or hubsets for sparse graphs, and especially for planar graphs, follow a particular line of thought, which can be called a \emph{counting argument}. One can construct a family of graphs $\mathcal{F}$ on $n$ vertices within the considered graph class, with a preselected subset of vertices $V' \subseteq V$, such that knowing $|V'|^2$ pairwise distances in graph $G$ is enough to identify the graph $G \in \mathcal{F}$. One can reason then that the total length of the $|V'|$ labels in some graph in $\mathcal{F}$ needs to be at least $\log_2 |\mathcal{F}|$, thus giving a lowerbound of $\frac{1}{|V'|} \log_2 |\mathcal{F}|$ bits per label (the same reasoning applies to hub-labeling schemes as well). This reasoning has one fundamental limitation, in that it is unable to separate \emph{a distributed data structure} from \emph{a centralized} one. For example, it is believed that unweighted planar graphs do in fact admit $\Omega(\sqrt{n})$ lowerbound for size of distance labels (as is the case for weighted planar graphs). Current lower bound construction of \cite{Gavoille:2004:DLG:1036161.1036165} uses counting technique construction with $|V'| = \Theta(n^{1/3})$ and $\log_2|\mathcal{F}| = \Theta(n^{2/3})$.  This cannot be improved with counting technique: \cite{DBLP:conf/soda/AbboudGMW18} shows that planar graphs do in fact admit \emph{distance oracles} which take in total $\bigo(\min{|V'|^2, \sqrt{n |V'|}})$ bits to encode all $|V'|^2$ pairwise distances, which is $\bigo(\min{|V'|, \sqrt{n/|V'|}})$ bits per vertex and is maximized when $|V'| = n^{1/3}$ (one needs to argue that counting technique is not able to separate distance oracles from distance labelings, which is the case). This motivates the search for alternative techniques when proving lower bounds.

\paragraph{Hub labeling in practice.} We remark that the hub-based method of distance computation is efficient for many real-world networks (even in centralised scenarios, as a distance oracle) for at least two reasons. First of all, for transportation-type networks it is possible to show bounds on the sizes of sets $S$, which follow from the network structure. Notably, Abraham et al.~\cite{doi:10.1137/1.9781611973075.64} introduce the notion of highway dimension $h$ of a network, which is presumed to be a small constant e.g.\ for road networks, and show that an appropriate cover of all shortest paths in the graph can be achieved using sets $S$ of size $\widetilde \bigo(h)$. Moreover, the order in which elements of sets $S(u)$ and $S(v)$ is browsed when performing the minimum operation is relevant, and in some schemes, the operation can be interrupted once it is certain that the minimum has been found, before probing all elements of the set. This is the principle of numerous heuristics for the exact shortest-path problem, such as contraction hierarchies and algorithms with arc flags~\cite{Kohler06fastpoint-to-point,Bauer:2010:SFR:1498698.1537599}.

\subsection{Overview of our results and proof techniques}

In this work we first show a lower bound on the size of hub labelings in sparse graphs. Note that in the regimes of both distance labeling and hub labeling, the state-of-the-art results for sparse graphs were leaving a huge gap: while upper bounds are of the form $\bigo(n/\log n)$ at best (ignoring poly-loglog terms), the lower bounds are $\Omega(\sqrt{n})$. We close this gap for hub labeling with the following result.

\newcommand{\ourlowerbound}{
Graphs of max-degree $3$ require average hub size to be at least $\frac{n}{2^{\Theta(\sqrt{\log n})}}$, where $n$ is the number of vertices.
}
\begin{theorem}
\label{th:ourlowerbound}
\ourlowerbound
\end{theorem}
(See Section~\ref{sec:lb} for a detailed proof.)

This result is complemented by an upper bound on hub set size, given by Theorem~\ref{th:ourupperbound}.

Our technical contributions are obtained by using the following observations as a starting point.
First, let us fix $D = n^{o(1)}$ to be a threshold value, such that we consider distances up to $D$ as small, and large otherwise. Simply using random hubsets of size roughly $n/D$ we can cover all $uv$ pairs with $uv$-distance being at least $D$. We therefore see immediately that only small distances are crucial (regardless of whether we are constructing a hard instance for a lowerbound, or small hubsets).

Next, we observe that in the regime of small distances, we can afford to fix our attention to \emph{monotone} hubsets, where we require that for any $u$, if $x \in S(u)$, then all vertices on some chosen shortest $ux$ path belong to $S(u)$ as well. It trivially follows that minimal  \emph{monotone} hubset covering all small distances is withing a factor of $D$ in size from minimal arbitrary hubset. Monotone hubsets have a following advantage in analysis: for $uv$ connected by an \emph{unique} shortest path, for any $x$ belonging to such path, either $u$ or $v$ has to ``pay'' for hub $x$, that is $x \in S(u)$ or $x \in S(v)$. Thus, looking for hard instances, it is advantageous to consider the following: fix $u \in V$, and consider set $S'(u) = \{ h : \textrm{dist}(u,h) = \textrm{dist}(h,v) = D/2 \}$ where $v \in V$ iterates over vertices such that $\textrm{dist}(u,v) = D$. If we ensure that $S'(u)$ need to be almost-linear in size (on average), since $\sum |S'(u)|$ is within a $D$ factor from total size of hubsets, we have our desired hard instance for hubset covering.

\paragraph{Induced matchings in dense graphs.}
It turns out that the connection between structure of shortest paths and induced matchings is at the heart of the hardness of the hub labeling problem in sparse graphs. This connection as such has already been noted in the literature (cf.~\cite{bodwin17}); here, we provide a brief exposition from a perspective most relevant to our study. Let us first recall the terminology. 

\begin{definition}
We say that $M \subseteq E(G)$ is an induced matching of $G$ if $(i)$ it is a matching and $(ii)$ there is a set of vertices $V' \subseteq V(G)$ such that $M$ is a subgraph of $G$ induced by $V'$.
\end{definition}
The connection between induced matchings and the shortest path structure follows from considering the bipartite graph $G' = (V,V,E')$, where $(u,v) \in E$ iff $u$ and $v$ are at distance at most $D$ in $G$, for some fixed threshold value $D$. Now consider a quadruple of vertices $u,v,u',v'$ such that $(u,v) \in G'$ and $(u',v') \in G'$, $\textrm{dist}(u,v) = \textrm{dist}(u',v') \le D$. Take now some  hub candidate $h \in V$ and some integers $a,b$ such that $a+b = \textrm{dist}(u,v)$. If $h$ is on the unique shortest $uv$ path and $h \in S(u) \cap S(v)$ and $\textrm{dist}(u,h) = a$, $\textrm{dist}(h,v) = b$, we can mark it by selecting $uv$ as an edge in some subgraph $M_h \subseteq E'$. We do the same for $u'v'$ and every other pair of vertices, adding edges to $M_h$ if $h$ is a hub on shortest path located at distance $a$ from one endpoint and $b$ from other.
It is now an easy observation that e.g. $\textrm{dist}(u,v') \le a+b$, and if the distance is indeed $a+b$, pair $u,v'$ is already covered by hub $h$. Since $h \in S(u)$ and $h \in S(v')$ was already amortized by paying for edges $uv$ and $u'v'$ in $M_h$, there is no need for adding edge $uv'$, and we can ensure that $M_h$ is indeed an induced matching (repeating the reasoning for $u'v$ pair).

We are thus interested in structure of (linear in $n$) induced matchings in (potentially dense) graphs. Specifically, when looking for hubset constructions, it is advantageous for us to have an upper bound on number of induced matchings graph can be edge partitioned into. Similarly, when looking for a lowerbound construction, we are interested in an explicit construction of dense graphs edge partitioned into large induced matchings, and moreover in having such graphs realisable as a structure of unique shortest paths. The first of these questions has been studied extensively in the combinatorics literature.

The Ruzsa-Szemer{\'e}di function $\RS(n)$ is a graph-theoretic function, defined as follows.
\begin{definition}
An undirected, unweighted graph $G$ on $n$ vertices whose edges can be partitioned into at most $n$ induced matchings is called a Ruzsa-Szemer{\'e}di graph.\footnote{In the literature there exists term of $(r,t)$-Ruzsa-Szemer{\'e}di graph, where edges are partitioned into $t$ induced matchings of size $r$. We follow the convention of \cite{bodwin17} and drop the indices.} Then $\RS(n)$ is the largest value such that every Ruzsa-Szemer{\'e}di graph has at most $n^2/\RS(n)$ edges.
\end{definition}
The study of values of $\RS(n)$ was initiated by \cite{ruzsa1978triple} who showed that $\RS(n) = \omega(1)$. Current bounds state that
$$2^{\Omega(\log^*n)} \leq \RS(n)  \leq 2^{\bigo(\sqrt{\log n})}$$
due to results of \cite{fox2011new} and \cite{behrend1946sets} (see \cite{DBLP:conf/soda/Elkin10} for current best upperbound up to lower-order terms). $\RS(n)$ finds applications in algorithms and complexity theory, e.g. communication over a shared channel \cite{DBLP:journals/tit/BirkLM93}, PCP Theorem \cite{DBLP:journals/rsa/HastadW03}, property testing \cite{DBLP:journals/jcss/AlonS04} and streaming algorithms \cite{DBLP:conf/soda/GoelKK12,DBLP:conf/soda/Kapralov13,DBLP:conf/esa/Konrad15}.
We note that another work to exploit the connections between $\RS(n)$ and the structure of shortest paths in graphs was  \cite{bodwin17}, although this was done for the different (distantly related) problem of constructing \emph{distance preservers}. Whereas the lite-motif of noting connections between sets of shortest paths and induced matchings is inherent to both~\cite{bodwin17} and our work, the combinatorial arguments used in the respective constructions are significantly different.

Using a number of combinatorial insights and properties of the hub labeling problem, we show the following result.

\newcommand{\ourupperbound}{
Any graph $G = (V,E)$ on $n$ vertices and $\bigo(n)$ edges admits a hub labeling $\{S_v\}$ of average size $\frac{1}{n} \sum_v |S(v)| = \bigo(\frac{n}{\RS(n)^{1/c}})$ for some constant $c \le 7$.
}
\begin{theorem}
\label{th:ourupperbound}
\ourupperbound
\end{theorem}
(See Section~\ref{sec:upperbound} for a detailed proof.)

This result can be read twofold. First, it hints at the possibility that there are hub labeling schemes for sparse graphs with average hub size $\bigo(\frac{n}{2^{\log^c n}})$ for some $0<c<1$, if one is to believe that the true value of $\RS(n)$ is on the side of the currently best upperbound. Secondly, it gives a (conditional) lowerbound on our lowerbound technique: one cannot hope to significantly improve the bound from Theorem~\ref{th:ourlowerbound} without strengthening existing bounds on value of $\RS(n)$, which is a hard open problem.

The connection between induced matchings and the structure of shortest paths leads us to the missing ingredient necessary for proving Theorem~\ref{th:ourlowerbound}. Namely, in hard instances we are looking for, in addition to already discussed structure, we need the structure of shortest paths to realize some form of Ruzsa-Szemer{\'e}di structure, that is admit edge partion into preferably $\bigo(n/\RS(n))$ induced matchings of $n$ edges. Luckily, related ideas for such constructions were explored by \cite{alon12}. The construction presented there involves graph with vertex set $[C]^d$ for some constants $C,d$, and interpreting vertices as $d$-dimensional vectors, the rule for connecting with edge $\vec{x}$ and $\vec{y}$ iff $|\|\vec{x}-\vec{y}\|_2^2 - \mu| \le n$ for some chosen constant $\mu$. We tweak the construction by using $[C]^d \times [\ell]$ as our vertex set, interpreted as $\ell$ $d$-dimensional layers, and connecting with an edge only between neighbouring layers (edges go always between $i$-th and $(i+1)$-th layer for some $i$). We fix $\ell$ to be a small value being in the order of \emph{diameter} of the graph in the original construction - thus there is an equivalence between paths in the original graph of length $\ell$ and path in the new graph from first to last layer. Finally, we simplify the predicate deciding the existence of an edge (edges connect $\vec{x}$ and $\vec{y}$ if they differ in at most one coordinate), and ensure the uniqueness of a shortest paths by introducing weights over the edges of the graph.

\paragraph{Distance labelings and and the Sum-Index problem.}
Our graph family used in the proof of Theorem~\ref{th:ourlowerbound} has a very regular structure. Namely, one can distinguish three (large) layers of vertices $A,B,C \subset V$, such that each layer is labeled with vectors from a high-dimensional integer grid, and the following property holds: for $u_{\vec x} \in A$ and $w_{\vec z} \in C$, the unique shortest path intersects $B$ in a single vertex that is precisely $v_{(\vec x + \vec z)/2} \in B$. It is tempting to construct a family of (sparse) graphs, where shortest path length from vertices in $A$ to vertices in $C$ is sensitive to existence (or removal) of vertices in $B$. We can imagine this scenario as a game, where two players, one residing on $A$ side, and one residing on $C$ side of the graph, are evaluating in coordination a $A \times C \to \{0,1\}$ function. To our advantage, such a problem has been widely considered in \emph{communication complexity} in the 1990's.

\begin{definition}[Sum-Index Problem]
Let $S = S_0S_1\ldots S_{n-1} \in \{0,1\}^n$. Alice holds $S$ and $a \in [0,n-1]$. Bob holds $S$ and $b \in [0,n-1]$. They both simultaneously send a messages $M_a$ and $M_b$ to a referee, whose goal is to compute $f(M_a,M_b) = S_{(a+b)\bmod n}$.
\end{definition}

This problem has been first stated explicitly in \cite{DBLP:conf/icalp/Pudlak94}, as a single-bit-output ``extract'' of a $\{0,1\}^n \times [n] \to \{0,1\}^n$ \textsf{shift} function: $\textsf{shift}_k(x) = y$, where $y_i = x_{(i+k) \bmod n}$. Analysis of $\textsf{shift}$ is closely connected to a program of proving highly non-trivial lowerbounds in circuit complexity, c.f. \cite{DBLP:journals/cc/HastadG91} for a precise reduction. Informally, the goal is to prove super-linear lowerbounds on the size of circuits for some function $\{0,1\}^n \to \{0,1\}^n$, and $\textsf{shift}$ was considered a strong candidate. However, the result of \cite{DBLP:conf/icalp/Pudlak94} and later \cite{ambainis96} have shown that such direction is unfeasible (at least with $\textsf{shift}$), by showing existence of sub-linear communication complexity protocols for evaluation of Sum-Index. Those results were called ''unexpected'' upperbounds, and the best construction up-to-date is due to \cite{ambainis96} with $\textsc{SumIndex}(n)= \bigo(\frac{n \log^{0.25} n}{2^{\sqrt{\log n}}})$ bits complexity, where $\textsc{SumIndex}(n)$ denotes the exact bit complexity of the problem.
 It is also known that $\textsc{SumIndex}(n) = \Omega(\sqrt{n})$, cf. \cite{babai03,DBLP:conf/stacs/BabaiKL95,DBLP:journals/siamcomp/PudlakRS97,DBLP:journals/siamcomp/NisanW93}, and the precise complexity is still a major open problem.
 We prove the following reduction.
 \newcommand{\communicationcomplexity}{Distance labeling in graphs on $n$ vertices and max-degree $3$ requires at least $\frac{1}{2^{\Theta(\sqrt{\log n})}}\textsc{SumIndex}(n)$ bits per vertex.}
 \begin{theorem}
 \label{th:communication_complexity}
 \communicationcomplexity
\end{theorem}
(See Section~\ref{sec:distancelabeling} for a detailed proof.)

We note that our linkage between distance labeling size and communication complexity lowerbounds appears to be the first advancement of the techniques in the area going beyond the graph-counting technique of \cite{Gavoille:2004:DLG:1036161.1036165}, and actually using distribution of information as a source of hardness in labeling schemes (although links between communication complexity and related to ours topic of compact routing have been investigated in \cite{twigg2006compact}). One might hope that, with the advancement of communication complexity bounds, this approach will eventually result in non-trivial \emph{unconditional} lowerbounds for distance labelings.
%We call $\{S_v \subseteq V : v \in V\}$ a hub labeling of an undirected graph $G = (V,E)$ if for all $u, v \in V$ there exists a node in $S_u \cap S_v$ which lies on some shortest $uv$-path in $G$. 

\section{Lower bound on Hub Labeling}
\label{sec:lb}
This section is devoted to the proof of the following theorem.

\begin{theorem}
\label{th:lowerbound_family}
For any positive integers  $\ell \in \N$ (the \emph{number of levels}) and $b \in \N$ (the \emph{side length parameter}) there exists a graph $G_{b,\ell}$ such that:
\begin{itemize}
\item[(i)] $|V(G_{b,\ell})| = 2^{b\ell} \cdot 2^{\Theta(b + \log \ell)}$ (number of nodes).
\item[(ii)] $\Delta(G_{b,\ell}) = 3$ (maximum degree).
\item[(iii)] any hub labeling $\{S_v\}$ of $G_{b,\ell}$ satisfies: $\frac{1}{|V(G_{b,\ell})|}\sum_{v \in V(G_{b,\ell})} |S_v| \geq  2^{b\ell} \cdot 2^{-\Theta(b+\ell)}$ (average size of hub sets).
\end{itemize}
\end{theorem}

%The following construction of $G_{b,\ell}$ is related to considerations of Ruzsa-Szemer{\'e}di-type decompositions (cf.~e.g.~\cite{bodwin17} for an overview of results in the area and their connection to distances in graphs). In particular, it can be compared to the geometric construction of almost-complete graphs decomposable into a linear number of induced matchings presented in~\cite{alon12}[Sec.\ 2.1].

\begin{proof}
%The claim made in the abstract follows for an appropriate choice of $d = \Theta(\eps^{-1} \log (\eps^{-1}))$.

In order to construct $G_{b,\ell}$, we more conveniently describe a weighted graph with non-uniform length edges, $H_{b,\ell} = (V, E, w)$ with integer edge weights appropriately chosen from the range $w(e) \in [1,(3\ell+1) \cdot 2^{2b}]$,\footnote{Throughout the text, we use the integer range notation $[a,b]\equiv \{a,a+1,\ldots,b\}$ and $[a,b) \equiv [a,b-1]$, for $a,b \in \N$.
} for all $e\in E$.

In following we set $s = 2^b$ (the \emph{side length}). We define vertex set of $H_{b,\ell}$ as $V = \bigcup_{r = 0}^{2\ell} V_i$, where each \emph{level} $V_i$, $i\in [0,2\ell]$, satisfies $|V_i| = s^\ell$ and is identified with a set of $\ell$-dimensional vectors, $V_i = \{v_{i,\vec j} : \vec j \in [0,s-1]^\ell\}$. %(It may help the intuition of the reader to put $d=1$; the one-dimensional case already captures the main properties of the construction --- see also Fig.~\ref{fig:oned}.) %By convention, we will treat levels $0$ and $2\ell$ specially and sometimes write $v_{\vec j} \equiv v_{0, \vec j}$ and $v'_{\vec j} \equiv v_{2\ell,\vec j}$.

The edges of $H_{b,\ell}$ are given so that we put an edge between $v_{i, \vec j}$ and $v_{i+1, \vec j'}$  when $\vec j$ and $\vec j'$ differ at most on one coordinate $c$: that is  $j_k = j'_k$ for all $k \not= c$, where $c = i+1$ for $i < \ell$ and $c = 2\ell-i$ for $i \ge \ell$.
The weight of of such an edge is then given as $w(\{v_{i, \vec j}, v_{i+1, \vec j'}\}) = A + (j_c - j'_c)^2$ where $A = 3\ell s^2$.
Note that each node $v_{i, \vec j}$ has exactly $s=2^b$ neighbors in $V_{i+1}$ (if $i<2\ell$) and $s$ neighbors in $V_{i-1}$ (if $i > 0$). See Figure~\ref{fig:H22} for an illustration.

\begin{figure}
    \centering
    \includegraphics[width=0.7\textwidth]{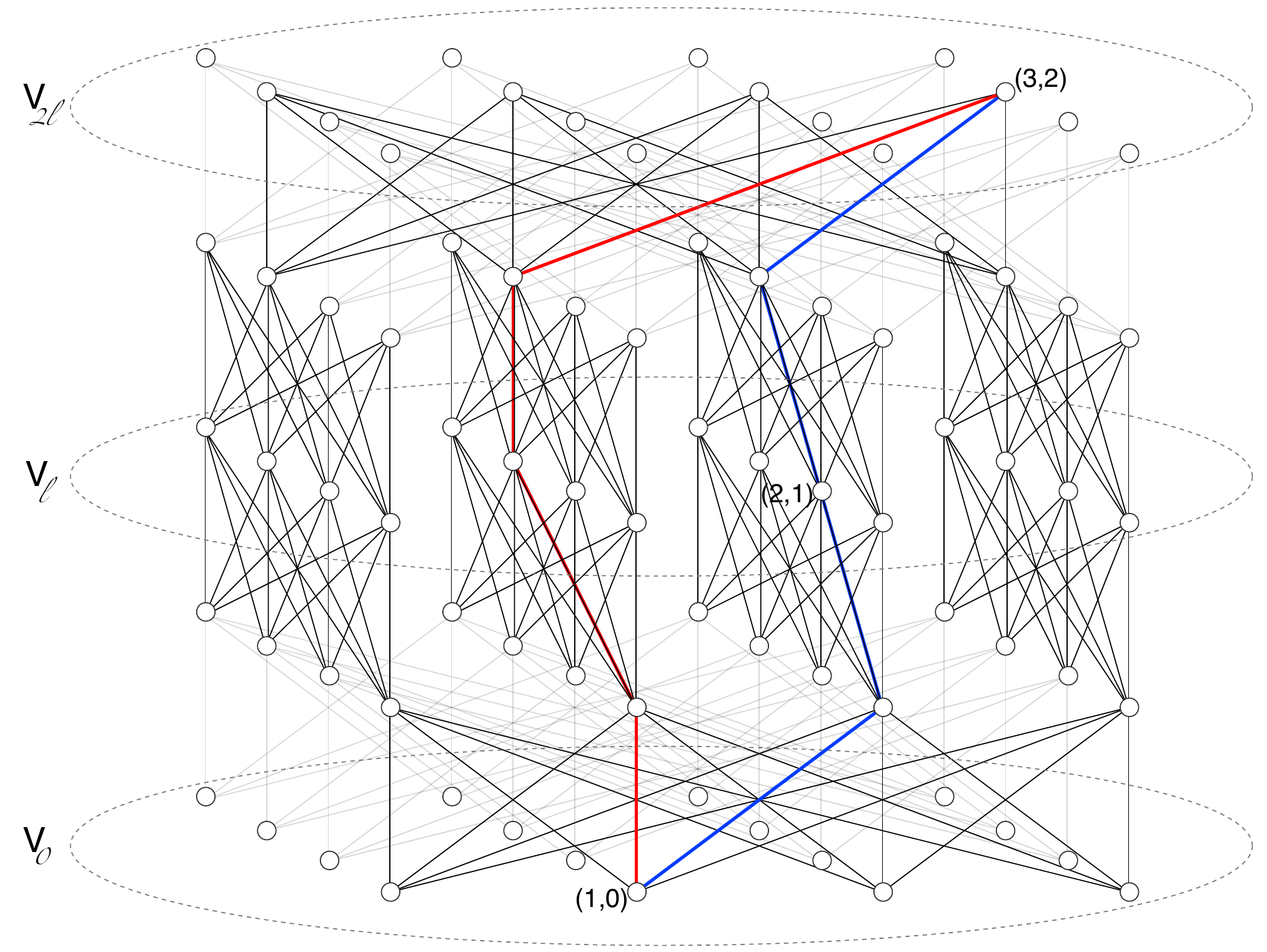}
    \caption{Graph $H_{b,\ell}$ with $b=2$ and $\ell=2$ ($s=4$). Some edges are drawn in light gray for better readability. The blue path is the only shortest path from $v_{0,(1,0)}$ to $v_{4,(3,2)}$. It passes through $v_{2,(2,1)}$ (which is a point of symmetry for the path) and has length $4A+4$. The red path has length $4A+8$.}
    \label{fig:H22}
\end{figure}

We convert $H_{b,\ell}$ into $G_{b,\ell}$ as follows.
\begin{itemize}
\item We associate each vertex $v$ of $H_{b,\ell}$ with two disjoint perfectly balanced binary trees $T^{\textrm{in}}_{v}$ and $T^{\textrm{out}}_{v}$ in $G_{b,\ell}$, both roots being linked to $v$. % (see also Fig. \ref{}). 
Each tree has $s$ leaves and depth $\log_2 s = b$. If $v \in V_i$, then leaves of $T^{\textrm{in}}_{v}$ are labeled as $v^{\textrm{in}}_u$, with $u \in V_{i-1}$, $\{u,v\} \in E(H_{b,\ell})$ ($T^{\textrm{in}}_{v}$ is omitted when $i=0$), and leaves of $T^{\textrm{out}}_{v}$ are labeled as $v^{\textrm{out}}_u$, with $u \in V_{i+1}$, $\{v,u\} \in E(H_{b,\ell})$ ($T^{\textrm{out}}_{v}$ is omitted when $i=2\ell$).
\item We associate each edge $e = \{u,v\}$ of $H_{b,\ell}$, $u \in V_i, v \in V_{i+1}$ of length $w(e)$ to a path of length $w(e)-2b+-2$ with $w(e)-2b-3$ auxiliary vertices in $G_{b,\ell}$, starting at $u^{\textrm{out}}_v$ and ending at $v^{\textrm{in}}_u$. Note that the path length satisfies $w(e)-2b+-2\ge 3\ell s^2-2b-2>0$ for $\ell\ge 1$ and $b\ge 1$, and that using $T^{\textrm{out}}_{u}$ and $T^{\textrm{in}}_{v}$, we obtain a path of length $w(\{u,v\})$ in $G_{b,\ell}$ for each edge $\{u,v\} \in E(H_{b,\ell})$.
\end{itemize}
We have $|V(G_{b,\ell})| \le |V(H_{b,\ell})| \cdot 4s + \sum_{e \in E(H_{b,\ell})} w(e) \le 4 s \cdot s^\ell \cdot (2\ell+1)  + (3\ell+1)s^2 \cdot s^\ell \cdot 2\ell \cdot s$, so the claims $(i)$ and $(ii)$ follow directly from the construction. To show claim $(iii)$, we first observe that the following property holds.

\begin{lemma}
\label{lem:pathlen}
Let $\vec x, \vec z \in [0,s-1]^\ell$ be such that for all $k \in [1,\ell]$, $z_k-x_k$ is even. Then there exists a unique shortest path in $H_{b,\ell}$ (and consequently also in $G_{b,\ell}$) between vertices $v_{0, \vec x}$ and $v_{2\ell, \vec z}$. Moreover, this shortest path passes through vertex $v_{\ell, (\vec x + \vec z)/2}$.
\end{lemma}
We defer the proof of the Lemma~\ref{lem:pathlen}.

Now, let $\{S_v\}$ be any fixed hub labeling of $G_{b,\ell}$. Fix arbitrarily shortest path trees $T_v$ rooted at each vertex $v$ of $G_{b,\ell}$, and let $S^*_v$ be the vertex set of the minimal subtree of $T_v$ rooted at $v$ which contains all vertices of $S_v$. We obviously have:
\begin{equation}\label{eq:star}
|S^*_v| \leq \diam(G_{b,\ell}) \cdot |S_v| \leq (3\ell + 1) s^2 \cdot 4\ell \cdot  |S_v|.
\end{equation}
Thus, the sets $S_v$ and $S^*_v$ are of equivalent size up to lower-order terms, and from now on, we focus on showing a lower bound on $\sum_{v\in V} S^*_v$. 

Consider all triplets $(\vec{x}, \vec{y}, \vec{z})$ such that $\vec{x},\vec{y},\vec{z} \in [0,s-1]^\ell$ and that $\vec{y} = (\vec{x}+\vec{z})/2$. There is $s^\ell \cdot (s/2)^\ell$ such triplets. Denote $x = v_{0,\vec{x}}$, $y = v_{\ell,\vec{y}}$ and $z = v_{2\ell,\vec{z}}$, then by Lemma~\ref{lem:pathlen} $y$ lies on the unique shortest $xz$ path, thus $y \in S^*_x$ or $y \in S^*_z$. Since $\vec{z} = 2\vec{x}-\vec{y}$, we have that value of $\vec{z}$ is uniquely determined by fixing $\vec{x}$ and $\vec{y}$, and similarly by $\vec{x} = 2\vec{z}-\vec{y}$, value of $\vec{x}$ is uniquely determined by fixing $\vec{y}$ and $\vec{z}$. Thus in each triplet $(\vec{x}, \vec{y}, \vec{z})$, vertex $y$ contributes by 1 to the size of $S^*_x$ or $S^*_z$. Thus $\sum_{v \in V} S^*_v \ge (s^\ell)^2 \cdot 2^{-\ell}$, and taking into account Eq.~\eqref{eq:star}, claim $(iii)$ follows.
\end{proof}

\begin{proof}[Proof of Lemma~\ref{lem:pathlen}]
Finally, we prove the Lemma~\ref{lem:pathlen}. We first consider the weighted graph $H_{b,\ell}$ only. Intuitively, we show that the unique shortest path in $H_{b,\ell}$ has the following structure: it climbs from level $0$ to level $\ell$ on its first $\ell$ edges, reaching vertex $v_{\ell, \vec y}$, and then climbs  next $\ell$ edges  to level $2\ell$; the portions of the path from $v_{0,\vec x}$ to $v_{\ell, \vec y}$ and from $v_{\ell, \vec y}$ to $v_{2\ell, \vec z}$ have a point symmetry with respect to $v_{\ell, \vec y}$. %See Fig.~\ref{fig:oned} for an illustration for the case of $d=1$.

Let $d(v_{0, \vec x},v_{2\ell, \vec z})$ denote the weighted distance between the pair of nodes considered in the Lemma. %For $k \in [1,d]$, we introduce notation for the bit representation of $|y_k - x_k|$:
%$$\frac{1}{2}|z_k - x_k| = |y_k - x_k| = |z_k - y_k| =: \sum_{i=0}^{\ell-1} 2^i b_{k,i},\quad \text{where $b_{k,i} \in \{0,1\}$}$$

Consider any path $P_{\vec x, \vec z} = (v_{\vec x} \equiv v_{0,\vec j^0}, v_{1,\vec j^1}, \ldots, v_{2\ell-1,\vec j^{2\ell-1}}, v_{2\ell,\vec j^{2\ell}} \equiv v_{\vec z})$ in $H_{b,\ell}$ on $2\ell$ edges. Denote $\delta_i = j^i_i-j^{i-1}_i$ and for $i \in [1,\ell]$ and $\delta_{i} = j^{i}_{2\ell-i+1} - j^{i-1}_{2\ell-i+1}$ for $i \in [\ell+1,2\ell]$.

% LV: this sounds like old stuff to be removed: 
%given as $j^{i}_{i} - j^{i-1}_{i} = j^{2\ell-i+1}_{i} - j^{2\ell-i}_{i} := \frac{1}{2} (z_{i} - x_{i})$, for $i \in [1,\ell]$.

The following properties hold:
\begin{itemize}
\item The weighted length of any such $P_{\vec x, \vec z}$ is:
$$
\sum_{i=1}^{2\ell}\left(A + \delta_i^2\right) \le 2\ell A + 2\ell s^2 < (2\ell+1)A.
$$
\item All $v_{0,\vec x} v_{2\ell,\vec z}$ paths in $H_{b,\ell}$ have at least $2\ell$ edges.
\item Any $v_{0,\vec x} v_{2\ell,\vec z}$ path in $H_{b,\ell}$ with at least $2\ell+1$ edges has weighted length at least $(2\ell+1)A$, which is greater than that of $P_{\vec x, \vec z}$.
\item $P_{\vec x, \vec z}$ satisfies $\delta_i + \delta_{2\ell+1-i} = z_i - x_i$. It thus follows that the total length $\sum_{i=1}^{\ell}\left(2A + \delta_i^2 + \delta_{2\ell+1-i}^2\right)$ is minimized when $\delta_i = \delta_{2\ell+1-i} = \frac{1}{2}(z_i-x_i)$, realized by an unique $v_{0,\vec x}v_{2\ell,\vec z}$ path having point symmetry at $v_{\ell,\vec y}$.
\end{itemize}
This finishes the proof for $H_{b,\ell}$. To complete the proof for $G_{b,\ell}$ we observe that for any $u\in V_i$ and $v\in V_j$ with $i<j$, we have $\textrm{dist}_{G_{b,\ell}}(u,v) = \textrm{dist}_{H_{b,\ell}}(u,v)$. This comes from the construction of $G_{b,\ell}$ for $j=i+1$ and from the fact that any $V_{i'}$ with $i <i'<j$ is a vertex cut in $G_{b,\ell}$ separating $u$ and $v$ for $j>i+1$: any shortest path $Q$ from $u$ to $v$ in $G_{b,\ell}$ must pass through $j-i-1$ vertices in $V_{i+1}\cup\cdots\cup V_{j-1}$ corresponding to a path $P$ in $G_{b,\ell}$ such that $|Q|=w(P)$. The shortest path from $v_{0,\vec{x}}$ to $v_{2\ell,\vec{z}}$ in $G_{b,\ell}$ is thus unique and passes through $v_{\ell,(\vec{x}+\vec{z})/2}$ similarly as the shortest path in $H_{b,\ell}$, which concludes the proof.
%
%following: (i) for any two $u,v$ such that $uv \in E(H_{b,\ell})$, we have $\textrm{dist}_{G_{b,\ell}}(u,v) = w(\{u,v\})$, (ii) for any $u,v \in V(H_{b,\ell})$, we have $\textrm{dist}_{G_{b,\ell}
%G_{b,\ell}}(u,v) = \textrm{dist}_{H_{b,\ell}}(u,v)$. Denoting by $Q_{\vec{x},\vec{z}}$ a natural embedding of $P_{\vec{x},\vec{z}}$ into $G_{b,\ell}$, by (i) we have $|Q_{\vec{x},\vec{z}}| = w(P_{\vec{x},\vec{z}})$.
%
%Denote by $Q'_{\vec{x},\vec{z}}$ some path from $v_{0,\vec{x}}$ to $v_{2\ell,\vec{z}}$ distinct from $Q_{\vec{x},\vec{z}}$. Since for $1 \le i \le 2\ell-1$, any of $V_i$ is a vertex cut in $G_{b,\ell}$ separating $v_{0,\vec{x}}$ and $v_{2\ell,\vec{z}}$ (this is true for any pair of vertices from $V_0$ and $V_{2\ell}$), $Q'_{\vec{x},\vec{z}}$ needs to pass through $2\ell-1$ distinct vertices from $V_1 \cup V_2 \cup \ldots \cup V_{2\ell-1}$, denote them $u_1,u_2, \ldots, u_{2\ell-1}$. Denote $u_0 = v_{0,\vec{x}}$ and $u_{2\ell} = v_{2\ell,\vec{z}}$. By $(ii)$, there is $|Q'_{\vec{x},\vec{z}}| \ge \sum_{0 \le i < 2\ell} \textrm{dist}_{G_{b,\ell}}(u_i,u_{i+1}) \ge 2\ell \cdot 2b + \sum_{0 \le i < 2\ell} \textrm{dist}_{H_{b,\ell}}(u_i,u_{i+1}) \ge 2\ell \cdot 2b + \textrm{dist}_{G_{b,\ell}}(u_i,u_{i+1}) > 2\ell \cdot 2b + w(P_{\vec{x},\vec{z}})$, which concludes the proof.
\qedhere
\end{proof}

Setting both $b = \ell = \sqrt{\log N}$ for some  appropriately chosen $N = n/2^{\Theta(\sqrt{\log n})}$ in Theorem~\ref{th:lowerbound_family} we finally obtain the main result of the Section. 

\mbox{}\newline\noindent\textbf{Theorem \ref*{th:ourlowerbound}.} 
\emph{
\ourlowerbound
}

\section{Lower bound for distance labeling}
\label{sec:distancelabeling}
We now show that distance labeling in sparse graphs is no easier than solving Sum-Index. We essentially use the graph construction from Theorem~\ref{th:ourlowerbound}, and the the fact that the distance between nodes $v_{0,\vec x}$ and $v_{2\ell, \vec z}$ is sensitive to the presence (or absence) in the graph of the node $v_{\ell, (\vec x + \vec z)/2}$. 

\mbox{}\newline\noindent\textbf{Theorem \ref*{th:communication_complexity}.}
\emph{
\communicationcomplexity
}

\begin{proof}
Let $S$ be the binary vector of length $m$ from Sum-Index Problem, with $m$ to be determined later. We describe the strategies of Alice and Bob of constructing the messages $M_a$ and $M_b$ respectively. Consider graph $G_{b,\ell}$ from Theorem~\ref{th:lowerbound_family}. We construct $G'_{b,\ell}$ by removing some vertices from layer $V_{\ell}$ (together with all adjacent edges) from $G_{b,\ell}$. We denote the choice of whether to include or remove particular vertex $v_{\ell,\vec{x}}$ as $W(\vec{x})$, and we defer how we decide those to later part of the proof.
\begin{observation}
\label{obs:decoding}
Let $\vec{x}$ and $\vec{z}$ be as in Lemma~\ref{lem:pathlen}, that is for any $k \in [1, \ell]$ $z_k - x_k$ is even. Then $W((\vec{x}+\vec{z})/2)$ can be decoded based only on $\vec{x}, \vec{z}$ and the length of the shortest path between $v_{0,\vec{x}}$ and $v_{2\ell,\vec{z}}$ in $G'_{b,\ell}$.
\end{observation}
We also note that $G'$ has $2^{b \ell} \cdot 2^{\Theta(b+\log \ell)}$ vertices and maximum degree of $3$.

Recall we denoted $s = 2^b$ the side-length of a layer in $G$ (or $G'$).
We fix $m = (s/2)^\ell$. For a vector $\vec{x} \in [s]^\ell$ we denote $\textrm{repr}(\vec{x}) = \left(\sum_i x_i \cdot (s/2)^i\right) \bmod m$ as the integer value coming from treating coordinates of $\vec{x}$ as digits in $(s/2)$-ary representation. Observe that while $\textrm{repr}()$ is a bijection between $[0,s/2-1]^\ell$ and $[0,(s/2)^\ell-1]$, this is not the case for the whole space $[0,s-1]^\ell$ (in fact every value is in the image of $2^\ell$ vectors). We fix the predicate $W(\vec{x})$ as: $W(\vec{x}) \equiv [S_{\textrm{repr}(\vec{x})} = 1]$.

The protocol for Alice is as follow:
\begin{enumerate}
    \item Alice constructs graph $G'_{b,\ell}$ based on binary word $S \in \{0,1\}^m$.
    \item Alice constructs a \emph{distance labeling} for $G'_{b,\ell}$.
    \item Alice finds unique $\vec{x} \in [0,s/2-1]^\ell$ such that $\textrm{repr}(\vec{x}) = a$ (based on the $(s/2)$-ary representation of $a$).
    \item Alice sends to the referee the label of vertex $v_{0,2\vec{x}}$ together with integer $a$.
\end{enumerate}
Bob proceeds analogously, representing $b$ as vector $\vec{z}$, constructing the same distance labeling of $G'_{b,\ell}$ as Alice and sending the label of $v_{2\ell, 2\vec{z}}$ with integer $b$.
The referee is then able to compute distance between Alice's and Bob's vertices, reconstructs vectors $\vec{x}$ and $\vec{z}$, and by the Observation~\ref{obs:decoding} knows the bit $S_{\textrm{repr}(\vec{x}+\vec{z})} = S_{(a+b) \bmod m}$, which follows from $\textrm{repr}(\vec{x}+\vec{z}) = (\textrm{repr}(\vec{x})+\textrm{repr}(\vec{z}) ) \bmod m = (a+b) \bmod m$.

We thus reach the conclusion that distance labeling of sparse graph with $2^{b\ell} \cdot 2^{\Theta(b + \log \ell)}$ requires at least $\textsc{SumIndex}(2^{(b-1) \ell} ) - b\ell$ bits per label. Fixing $b = \ell = \sqrt{\log N}$ for some $N = \frac{n}{2^{\Theta(\sqrt{\log n})}}$ finishes the proof.
\end{proof}

\section{Upper bound}
\label{sec:upperbound}
In this section we show how to take advantage of a structure of induced matchings in dense graphs to construct hubsets. We start by considering a case of graphs of constant maxdegree.
\begin{theorem}
\label{th:upperbound}
Any graph $G = (V,E)$ on $n$ vertices and maxdegree $\Delta = \bigo(1)$ admits a hub labeling $\{S_v\}$ of total size $\sum_v |S_v| = \bigo(\frac{n^2}{\RS(n)^{1/c}})$ for some constant $c \le 7$.
\end{theorem}
\begin{proof}
Let $D$ be a parameter $1 \le D \le \RS(n)$ to be fixed later. For any pair of vertices $u,v$ we denote $H_{uv}$ the set of all valid hubs for this pair, that is $H_{uv} \equiv \{x : \textrm{dist}(u,x)+\textrm{dist}(x,v) = \textrm{dist}(u,v)\}$.

We start by formulating the following property for a set $S \subseteq V$: $(\ast)$ for two vertices $u,v$ such that $|H_{uv}| \ge D$, we have $S \cap H_{uv} \not= \emptyset$. There exists a set $S$ such that $|S| = \bigo(\frac{n}{D} \log D)$ and $(\ast)$ is satisfied for all pairs $uv$ except $\frac{n^2}{D}$ many.
Existence of such set can proven by probabilistic method (c.f. \cite{BCE05}). Indeed, pick such set uniformly at random with $|S| = \frac{n}{D} \ln D$. Any pair $u,v$ with at least $D$ hubs is \emph{not covered} with probability at most $(1-D/n)^{|S|} \le 1/D$, thus the expected number of uncovered pairs is at most $n^2/D$, and there is a selection of $S$ with at most that many pairs not covered. From now on we focus on such $S$, and we denote by $\{Q_v\}$ a family of sets such that if a pair $uv$ is not covered by $S$ and satisfies $H_{uv} \ge D$, then we put $v \in Q_u$. By the property of $S$, we have $\sum_{v \in V} |Q_v| \le n^2/D$.

We now color all vertices of $G$ with $D^3$ colors, so that for each $v \in V$, the color $c_v \in [1,D^3]$ is assigned independently and uniformly at random.  Consider the following event: for a pair $u,v$, each vertex from $H_{uv}$ is assigned a different color. If $|H_{uv}| \le D$, then such an event happens with probability at least $1-1/D$, and the opposite event with probability at most $1/D$. If we define sets $\{R_u\} = \{ v : \text{ there are } x,y \in H_{u,v} \text{ such that } c_x = c_y \}$, then by simple computation of expected value, there is a choice of colors $c_v$ such that $\sum_v |R_v| \le n^2/D$. Thus we can afford for each $v$ to store as hubs $R_v$, the vertices where coloring of potential hubs failed to assign unique colors.

We now deal with $u,v$ such that $|H_{uv}| \le D$ and $H_{uv}$ was properly colored using different colors.
First, we iterate through $a,b \ge 0$ such that $1 \le a+b \le D$ and iterate $h \in V$. Consider bipartite graph $(V,V,E_{a,b}^h)$,  $E_{a,b}^h \subseteq V \times V$. For $u,v$ we put $(u,v) \in E_{a,b}^h$ if the following conditions hold: $(i)$ $|H_{uv}| \le D$, $(ii)$ each vertex of $H_{uv}$ was colored using distinct color, $(iii)$ $h \in H_{uv}$, $\textrm{dist}(u,h) = a$ and $\textrm{dist}(h,v) = b$. 
We now use the following Lemma, with proof provided later:

\begin{lemma}
\label{lem:rsbound}
Construct $\{F_v\}$ as follow: $a,b$ iterate so that $1 \le a+b \le D$ and $h$ iterate over $V$, and consider some minimum vertex cover $(V_1,V_2), V_1,V_2 \subseteq V$ of $(V,V,E_{a,b}^h)$.  If $v \in V_1 \cup V_2$, we put $h$ into $F_v$. We then have $\sum_v |F_v| = \bigo(D^5 \frac{n^2}{\RS(n)}).$
\end{lemma}

For $X \subseteq V$, let $N(X) = \{ v : \exists_{x \in X}  \textrm{dist}(v,x) \le 1 \}$ denote neighborhood. %of diameter 1 of $X$.
We observe, that for any two vertices $u,v$, one of the following holds:
\begin{enumerate}
\item If $H_{uv} \ge D$, then either there exists $h \in S$ such that $h \in H_{uv}$ is a hub for $u,v$, or $v \in Q_u$ is a hub for $u,v$.
\item If $H_{uv} \le D$ and there is a color conflict in $H_{uv}$, then $v \in R_u$ is a hub for $u,v$.
\item If $H_{uv} \le D$ and there is no color conflict in $H_{uv}$, then for any $h \in H_{uv}$, there are $a = \textrm{dist}(u,h)$ and $b = \textrm{dist}(h,v)$ such that $(u,v) \in E_{a,b}^h$. Thus at least one of $u,v$ is in the corresponding vertex cover, and so $h \in F_u$ or $h \in F_v$. Since this holds for any $h \in H_{uv}$, and w.l.o.g. $u \in F_u$, $v \in F_v$, by induction along a $uv$-shortest path there is an edge $(x,y) \in E$ such that $x \in F_u$ and $y \in F_v$ and $x,y \in H_{uv}$, guaranteeing that $N(F_u) \cap N(F_v) \not= \emptyset$.
\end{enumerate}
We then set each vertex hubset $H_v = S \cup Q_v \cup R_v \cup N(F_v)$. A bound on size follows
$\sum_v |H_v| \le n \cdot |S| + \sum_v |Q_v| + \sum_v |R_v| + (\Delta+1) \sum_v |F_v| \le \bigo(\frac{n^2}{D} \log D) + \bigo(\frac{n^2}{D}) + \bigo(D^5 \frac{n^2}{\RS(n)})$ using Lemma~\ref{lem:rsbound}. Setting $D = \RS(n)^{1/6}$ completes the proof.

\end{proof}

\begin{proof}[Proof of Lemma~\ref{lem:rsbound}.]
Fix values $a,b$ and $h$ and vertex cover $VC_{a,b}^h$ used in that iteration. By relation between maximum matching and minimum vertex cover, any maximal matching $MM_{a,b}^h \subseteq E_{a,b}^h$ satisfies $|VC_{a,b}^h| \le 2 |MM_{a,b}^h|$. For each color $c$, consider bipartite graph $G_{a,b}^c = \bigcup_{h : c_h = c} MM_{a,b}^h$. We now show that $MM_{a,b}^h$ is in fact an induced matching in $G_{a,b}^c$. It is enough to show the following:  for any two distinct $(u_1,v_1), (u_2,v_2) \in MM_{a,b}^h$, we have $(u_1,v_2) \not\in G_{a,b}^c$ and $(u_2,v_1) \not\in G_{a,b}^c$. Assume otherwise, that $(u_1,v_2) \in MM_{a,b}^{h'}$ for some $h' \not= h$ (we know we can exclude $(u_1,v_2) \in MM_{a,b}^{h}$ due to $MM_{a,b}^{h}$ being a matching). We know that $\textrm{dist}(u_1,h) = a$ and $\textrm{dist}(h,v_2) = b$, and each vertex of $H_{u_1v_2}$ was colored using a different color, thus either $\textrm{dist}(u_1,v_2) < a+b$ or $c_{h'} \not= c_h$, a contradiction.

This shows that $G^c_{a,b}$ is a Ruzsa-Szemer{\'e}di graph (on $2n$ vertices), with edge partition into induced matchings $MM_{a,b}^h$ for $h : c_h = c$. It thus follows that $|G^c_{a,b} |\le \frac{(2n)^2}{\RS(2 n)} = \bigo(\frac{n^2}{\RS(n)})$.\footnote{It is straightforward to show $\RS(2n) = \Theta(\RS(n))$.} We then obtain the following bound:
\begin{align*}
    \sum_{v \in V} |F_v| &\le \sum_{1 \le a+b \le D} \sum_{h \in V} |VC_{a,b}^h|\\ 
    &\le \sum_{1 \le a+b \le D} \sum_{h \in V} 2|MM_{a,b}^h|\\ 
    &\le \sum_{1 \le a+b \le D} \sum_{c \in [1,D^3]} 2 |G_{a,b}^c|\\ 
    &= \bigo(D^5 \frac{n^2}{\RS(n)}). 
\end{align*}
\end{proof}

Note that the construction used in proof of Theorem~\ref{th:upperbound} generalizes to edges with $\{0,1\}$ weights since any path length $c \leq D$ still decomposes in at most $D$ sums $c=a+b$ of two path lengths.
We now note that if $G$ is a graph of constant \emph{average degree}, that is $m/n = \bigo(1)$, then we can reduce construction of hub labeling to constant \emph{max degree} case as follow. First, we construct $G'$ by subdividing any vertex $v \in V(G)$ of degree $\textrm{deg}(v)$ into $\lceil \frac{\textrm{deg}(v)}{\lceil m/n \rceil} \rceil$ vertices of degree at most $2+\lceil m/n \rceil$, using a path with weight-0 auxiliary edges for linking them (treating all the non-auxiliary edges as weight-1 edges). We have then $|E(G')| = \bigo(m)$ and $|V(G')| = \bigo(m)$. We note that for any $v' \in V(G')$ there is $v \in V(G)$ that originated $v'$, and for any $v \in V(G)$ we can pick one $v' \in V(G')$ as a representative of $v$ in $G'$. Having constructed hub labeling for $G'$, denote it $\{H'_v\}$, we construct a hubset of $v \in V(G)$ by taking hubset of its representative in $G'$ and projecting back each selected hub in $G'$ to its original vertex in $G$. We thus have reached our main result on the side of upper bounds.

\mbox{}\newline\noindent\textbf{Theorem \ref*{th:ourupperbound}.}
\emph{
\ourupperbound
}

%\clearpage

\bibliographystyle{alpha}
\bibliography{bib}
%\paragraph*{Acknowledgment. The authors gratefully thank Pawe\l{} Gawrychowski for pointing us to the formulation of the SumIndex problem.}
\end{document}